\documentclass[12pt, letterpaper]{JHEP3}

\usepackage{amssymb, amsmath, amsopn, amsthm}
\usepackage{epsfig}

\newcommand{\cN}{\mathcal{N}}

\newcommand{\cJ}{\mathcal{J}}

\newcommand{\cL}{\mathcal{L}}

\newcommand{\cT}{\mathcal{T}}

\newcommand{\beq}{\begin{equation}}
\newcommand{\eeq}{\end{equation}}

\usepackage{bbm}
\DeclareFontFamily{U}{rsf}{}
\DeclareFontShape{U}{rsf}{m}{n}{
  <5> <6> rsfs5 <7> <8> <9> rsfs7 <10-> rsfs10}{}
\DeclareMathAlphabet\Scr{U}{rsf}{m}{n}

\def\CO#1#2{{[#1,#2]}}

\def\iden{{\mathbbm 1}}

\def\rep#1{{{\boldsymbol{#1}}}}
\def\brep#1{{{\overline{\boldsymbol{#1}}}}}

\def\C{{\mathbb C}}

\def\Sym{\operatorname{Sym}}

\def\GL{\operatorname{GL}}

\def\SU{\operatorname{SU}}
\def\GU{\operatorname{U{}}}
\def\Sp{\operatorname{Sp}}

\def\GE{\operatorname{E}}

\def\so{\operatorname{\mathfrak{so}}}

\def\sp{\operatorname{\mathfrak{sp}}}

\def\su{\operatorname{\mathfrak{su}}}
\def\Lu{\operatorname{\mathfrak{u}}}
\def\Le{\operatorname{\mathfrak{e}}}
\def\Lg{\operatorname{\mathfrak{g}}}
\def\LLh{\operatorname{\mathfrak{h}}}

\def\cJ{{\cal J}}

\def\cL{{\cal L}}

\def\cN{{\cal N}}

\def\cT{{\cal T}}

\newcommand\psit{\widetilde{\psi}}

\newcommand\qt{\tilde{q}}

\newcommand\Qt{\widetilde{Q}}

\newtheorem{theorem}{Theorem}
\newtheorem{lemma}{Lemma}

\def\tadj{\operatorname{adj}}

\def\LCg{{{\mathfrak{C}_{\Lg}}}}

\newcommand{\bQ}{{\bf Q}}
\newcommand{\bM}{{\bf M}}
\newcommand{\uu}{{\mathfrak u}}
\newcommand{\eu}[1]{{\mathfrak #1}}

\title{Global Symmetries and ${\cal N}=2$ SUSY}

\author{Jock McOrist${}^1$, Ilarion V.~Melnikov${}^2$, and Brian Wecht${}^3$\\

{\tt j.mcorist@surrey.ac.uk, ilarion@physics.tamu.edu, b.wecht@qmul.ac.uk}\\
${}^1$ Department of Mathematics, University of Surrey \\ 
Guildford, Surrey, GU2 7XH, UK \\
${}^2$ George P. and Cynthia W. Mitchell Institute for Fundamental Physics and Astronomy,
Texas A\&M University \\
College Station, TX 77843, USA \\
${}^3$ Centre for Research in String Theory, Queen Mary University of London \\ 
London E1 4NS, UK \\
}

\abstract{We prove that $\cN=2$ theories that arise by taking $n$ free hypermultiplets and gauging a subgroup of $\Sp(n)$, the non-R global symmetry of the free theory, have a remaining global symmetry which is a direct sum of unitary, symplectic, and special orthogonal factors. This implies that theories that have $\SU(N)$ but not $\GU(N)$ global symmetries, such as Gaiotto's $T_N$ theories, are not likely to arise as IR fixed points of RG flows from weakly coupled ${\cN=2}$ gauge theories.}

\preprint{DMUS-MP-13/20 \\ MIFPA 13-34 \\ QMUL-PH-13-13}

\begin{document}
\section{Introduction}
\label{sec:intro}

Classifying the different possible phases of quantum field theories has been a longstanding goal of high energy theoretical physics, and understanding and constraining the symmetries that arise in particular realizations is a key tool in this effort.  In some cases, such as in two dimensions, there has been a significant amount of progress in this direction, {\it e.g.,} the known restriction of unitary conformal field theories (CFTs) with $c<1$ to the minimal models, where the chiral algebra essentially fixes the theories. In four dimensions, however, significantly less is known, even in the case of CFTs.

It has long been known that it is possible to engineer four-dimensional CFTs which do not obviously have any free-field limit. An early class of examples are the $\cN=2$ SCFTs found by Minahan-Nemeschansky \cite{Minahan:1996cj,Minahan:1996fg}. These theories have $\GE_{6,7,8}$ global symmetries, and can be studied via the Seiberg-Witten curve \cite{Seiberg:1994rs, Seiberg:1994aj} and the powerful techniques available in $\cN=2$ theories. Although much is known about these theories, including the dimensions of various operators, 't Hooft anomalies, and even some chiral ring relations \cite{Gaiotto:2008nz}, there is no known way of directly constructing the theories via an asymptotically free UV theory.\footnote{It is worth noting that these theories, albeit with certain global symmetries gauged, can be realized via Argyres-Seiberg duality and generalizations \cite{Argyres:2007cn, Argyres:2007tq}. However, much like in the case of Argyres-Douglas theories \cite{Argyres:1995jj}, there is not a straightforward mapping between the weakly coupled degrees of freedom and those of the ungauged $\GE_n$ theories. } Shortly after the discovery of Argyres-Seiberg duality, it was realized  \cite{Gaiotto:2009we} that the Minahan-Nemeschansky CFTs are in fact special cases of a much broader class of $\cN=2$ theories that come from wrapping M5-branes on a three-punctured sphere. The $\GE_6$ theory is a special case of Gaiotto's $T_N$ theories \cite{Gaiotto:2009we}, and $\GE_{7,8}$ are special cases that emerge when allowing more general punctures on the sphere \cite{Tachikawa:2009rb,Chacaltana:2010ks,Chacaltana:2011ze}. For all but a few very special cases, which are free theories, these theories do not have known UV Lagrangian descriptions.   Needless to say, such a description could be of great use---for instance, one could apply powerful localization techniques to constrain and perhaps fix the chiral ring structure of a given theory.  This leads to a natural question:  are there theories for which we can rule out the existence of a useful Lagrangian formulation?\footnote{By utility we mean that the connection between UV and IR physics is relatively simple, ideally without the complications of a strong coupling limit or accidental symmetries.}

Despite the lack of a Lagrangian description, it is still possible to do detailed calculations in these theories. This is because for many quantities of interest, knowing information about the global symmetries such as the leading behavior of current two- and three-point OPEs is sufficient, and global symmetry currents are among the limited set of operators to which we have reliable access. Although useful in general, global symmetry information has proved particularly important for studying $\cN=1$ generalizations of the $T_N$ theories, as in \cite{Benini:2009mz} and subsequent work. This brings up the general question of what sorts of constraints follow from the global symmetries of these theories.

In this work we make the observation that these two questions, {\it i.e.} the constraints on possible symmetries and existence of a Lagrangian, have an interesting relation in the context of $\cN=2$ gauge theories. We will show that some (non-R) global symmetries, such as the $\SU(N)^3$ global symmetry possessed by Gaiotto's $T_N$ theories, are not straightforwardly realized by asymptotically free $\cN=2$ theories.  The essence of our argument is that such $\SU(N)$ symmetries are always accompanied by an additional $\GU(1)$ which enhances the symmetry to $\GU(N)$. Although we will not be able to completely rule out the possibility that the $T_N$ theory has a UV Lagrangian description, we will be able to place constraints on any gauge theory realization.  We will discuss these constraints and their limitations further in section~\ref{sec:disc}.

The main result of our paper is a proof that the global symmetries of certain $\cN=2$ gauge theories fall into a straightforward classification depending on the matter representation. Our starting point will be a theory of $n$ free hypermultiplets, which has a non-R global symmetry group $\Sp(n)$. We prove that after gauging a subalgebra $\Lg$ of the global symmetry algebra $\sp(n)$, the remaining global symmetry algebra is a direct sum of $\so, \sp$, and $\uu$ factors. In particular, we note that $\su$ factors without accompanying $\uu(1)$'s do not appear. This classification is certainly known to some experts (see for example \cite{Argyres:2007tq, Argyres:1996eh}), but we are not aware of a general proof in the literature. Our aim is to provide such a proof and explore some of the consequences.

\section{Symmetries of free fields}
\label{sec:free}

It is instructive to first understand the global symmetry of a theory of $n$ free hypermultiplets.  
In $\cN=1$ superspace a hypermultiplet consists of a chiral superfield $Q$ with propagating component fields $(q,\psi)$, and a chiral superfield $\Qt$ with components $(\qt,\psit)$. Requiring  $\cN=2$ supersymmetry implies there is a  $\GU(1)_{\text{R}} \times \SU(2)_{\text{R}}$ R-symmetry, 
under which $(q,\qt^\dag)$ transform as a doublet under $\SU(2)_{\text{R}}$, while the fermions are neutral. We parametrize the $\SU(2)_{\text{R}}$ action on the bosons as
\beq
\begin{split}
 T_{R} : \begin{pmatrix} q \\ \qt \end{pmatrix} \mapsto \begin{pmatrix} a q + b \qt^\dag \\ -b q^\dag + a \qt \end{pmatrix}~,\qquad |a|^2 + |b|^2 = 1.
\end{split}
\eeq
In what follows we split the $2n$ chiral multiplets into a column vector $Q$ and a row vector $\Qt$ (with transpose $\Qt^t$), so that
the Lagrangian for $n$ free hypermultiplets is
\begin{align}
\label{FreeFieldTheory}
\cL =  \int d^4\theta ~ \bQ^\dag \bQ~,\quad 
\bQ \equiv \begin{pmatrix} Q \\ \Qt^{t} 
\end{pmatrix}.
\end{align}
We want to identify global symmetries that commute with both ${\cal N}=1$ and $\SU(2)_{\text{R}}$. The first requirement means that these global symmetries must act linearly on the superfields $\bQ$:
\begin{align}
T_{\bf M}: \bQ \to \bM \bQ,\qquad  \quad
T_{\bf M}: \begin{pmatrix} Q \\ \Qt^{t}
\end{pmatrix} \to
\begin{pmatrix} M_1 & N_1 \\ N_2 & M_2
\end{pmatrix}
\begin{pmatrix} Q \\ \Qt^{t}
\end{pmatrix}~,
\end{align}
where ${\bf M}$ satisfies ${\bf M}{\bf M}^\dag = \iden_{2n}$, {\it i.e.}  $\bM \in \GU(2n)$.  Since the $\SU(2)_{\text{R}}$ acts trivially on fermions, we just need to determine the set of ${\bf M}$ restricted to the bosons that commute with the $\SU(2)_{\text{R}}$ action. Evaluating the composition of two arbitrary rotations on the chiral fields explicitly,
\begin{align}
T_{R} T_{\bM} &: \begin{pmatrix} q \\ \qt \end{pmatrix} \mapsto
\begin{pmatrix}
a(M_1 q + N_1 \qt^t) + b (N_2^\ast q^\ast + M_2^\ast \qt^\dag) \\
-b(\qt^\ast N_1^\dag + q^\dag M_1^\dag) + a (q^t N_2^t + \qt M_2^t)
\end{pmatrix}~,
\nonumber\\
T_{\bM} T_{R} & : \begin{pmatrix} q \\ \qt \end{pmatrix} \mapsto
\begin{pmatrix}
a (M_1 q + N_1 \qt^t) + b (M_1 \qt^\dag -N_1 q^\ast) \\
a(q^t N_2^t + \qt M_2^t) + b(\qt^\ast N_2^t - q^\dag M_2^t)
\end{pmatrix}~,
\end{align}
we see that $[T_{\bf M},T_R]=0$ if and only if
\begin{align}
M_1 = M_2^\ast,\qquad N_1 = - N_2^\ast~.\label{SpCondition}
\end{align}
Equivalently,  $\bM J \bM^t = J$, where $J$ is the symplectic structure
%
\begin{align}
J = \begin{pmatrix}  0 & \iden_{n} \\ -\iden_{n}  & 0 
\end{pmatrix}~.
\end{align}
Hence $\bM \in \GU(2n)\cap \Sp(2n,\C) \equiv \Sp(n)$, the compact unitary symplectic group.\footnote{In these conventions  $\Sp(1) = \SU(2)$.}
%
%
 We have uncovered the global symmetry group of $n$ free hypermultiplets: $\GU(1)_{\text{R}}\times \SU(2)_{\text{R}} \times \Sp(n)$, with matter in the fundamental of $\Sp(n)$, a pseudoreal representation.\footnote{As discussed in~\cite{Seiberg:1994aj}, at the level of groups this action is not completely disjoint from that of the Lorentz group, but that will not affect our analysis at the level of the algebra.} 

\section{Representation theory}
\label{sec:reps}
In this section we will characterize the global symmery algebra of a weakly coupled Lagrangian $\cN=2$ gauge theory.  Starting with a free theory of $n$ hypermultiplets, we gauge a semisimple subalgebra $\Lg$ of the global symmetry algebra $\sp(n)$ of the free theory.  The global symmetry algebra $\LCg$ is the commutant of $\Lg$ in $\sp(n)$, i.e. 
\begin{align}
\LCg = \{ x \in \sp(n) ~|~~ \CO{x}{y} = 0\quad\text{for all}~ y \in \Lg\}~.
\end{align}
This is also known as the centralizer of $\Lg$ in $\sp(n)$. We will prove the following theorem.
\begin{theorem}
\label{thm:fieldsmedal}
Let $\Lg$ be a semisimple subalgebra of $\sp(n)$.  Then the commutant subalgebra $\LCg$ of $\Lg$ in $\sp(n)$ is
\begin{align*}
\LCg = \bigoplus_i \sp(k_i) \oplus \bigoplus_p \so(l_p) \oplus \bigoplus_q \Lu(m_q)~,
\end{align*}
and the fundamental of $\sp(n)$ decomposes under $\sp(n) \supset \Lg \oplus\, \LCg$ as
\begin{align*}
\rep{2n} = 
\bigoplus_i (\rep{r^+_i},\rep{2k_i})\oplus
\bigoplus_p (\rep{r^-_p},\rep{l_p}) \oplus
\bigoplus_q  \left[ (\rep{r^c_q}, \rep{m_q}) \oplus (\brep{r^c_q},\brep{m_q})\right]~,
\end{align*}
where $\rep{r^+_i}$, $\rep{r^-_p}$, $\rep{r^c_q}$ are distinct irreducible representations of $\Lg$ that are, respectively, real, pseudoreal, or complex, and $\rep{2k_i}$, $\rep{l_p}$, and $\rep{m_q}$ denote the fundamental representations of the corresponding factors in $\LCg$.
\end{theorem}

The result has a simple implication for the physics:  if we gauge a semisimple $\Lg \subset \sp(n)$, then the global symmetry group will be a sum of classical Lie algebras acting on the different flavors in fundamental representations.

\subsection{A few familiar gaugings}
Before we turn to the general case we will review the familiar cases of $\cN=2$ SQCD with $\Lg$ one of $\su(p)$, $\sp(q)$, or $\so(m)$~\cite{Argyres:1996eh,Argyres:1996hc}.  This is accomplished via the embeddings
\begin{align}
\label{eq:sqcd}
\sp(pm) &\supset \su(p)\oplus\Lu(m)~,& \rep{2pm} &= (\rep{p},\rep{m})\oplus(\brep{p},\brep{m})~,\nonumber\\
\sp(qm) &\supset \sp(q) \oplus \so(m)~,& \rep{2qm} & = (\rep{2q},\rep{m})~.
\end{align}
It is  straightforward to then construct embeddings for any simple $\Lg \subset \sp(n)$.  Suppose $\rep{r}$ is an irreducible representation (irrep) of $\Lg$ of dimension $k$.  Then, depending on whether $\rep{r}$ is real, pseudoreal, or complex, there is an S-subalgebra embedding $\Lg \subset \so(k)$, $\Lg \subset \sp(k)$, or $\Lg \subset \su(k)$~\cite{Cahn:1984la}.  It is then a simple matter to use the embeddings in~(\ref{eq:sqcd}) to construct suitable gauge theories.  For instance, to build a $\Le_6$ gauge theory with $s$ hypermultiplets in the $\rep{27}$, we need $s$ conjugate multiplets in $\brep{27}$, and we use the embedding
\begin{align}
\sp(27 s) &\supset \su(27) \oplus \Lu(s) \supset \Le_6 \oplus \Lu(s)~,&
\rep{54 s} = (\rep{27},\rep{s}) \oplus (\brep{27},\brep{s})~.
\end{align}
In all of these cases the reality properties of various irreps play a key role in constructing the embedding. As we will see this will also be the case more generally.  Our strategy will rely on two simple facts:
\begin{enumerate}
\item the decomposition of $\rep{2n}$ under $\sp(n)\supset \Lg \oplus\, \LCg$ determines the decomposition of $\tadj \sp(n) = \Sym^2\!\rep{2n}$;
\vspace{-2mm}
\item $\rep{2n}$ is usefully decomposed according to reality properties of irreps of $\LLh$. 
\end{enumerate}

\subsection{Warm-up:  decomposing pseudoreal representations}
We begin by fixing some useful conventions and reviewing a few definitions and familiar facts from representation theory.  Throughout we work  with anti-Hermitian generators $\cT$ for the Lie algebras.  The standard definitions for real/pseudoreal/complex representations are then as follows~\cite{Cahn:1984la,McKay:1981rt,Slansky:1981yr}.  Let $\Lg$ be a simple Lie algebra with irrep $\rep{r}$.   Schur's lemma and some familiar facts about complex matrices~\cite{Zumino:1962nf} imply that, up to a change of basis, $\rep{r}$ admits at most one bilinear invariant, which must either be symmetric or skew-symmetric. This leads to a classification of the irreps as either real, pseudoreal or complex, which we will denote by superscripts $\rep{r^+}$, $\rep{r^-}$, and $\rep{r^c}$: \\[-4mm]
\begin{enumerate}
\item
$\rep{r}$ is real if $\Sym^2\!\rep{r} \supset \rep{1}$.  We can choose a basis for $\rep{r}$ so that $\cT_{\rep{r}} = \cT^\ast_{\rep{r}}$ are real skew-symmetric matrices, so that if $\rep{r}\neq \rep{1}$, then $\wedge^2 \rep{r} \supset \tadj\Lg$.
\item
$\rep{r}$ is pseudoreal if $\wedge^2 \rep{r} \supset \rep{1}$.  In this case $\dim\rep{r}$ is even, and we can choose a basis for $\rep{r}$ so that $\cT^\ast_{\rep{r}} = -\cJ \cT_{\rep{r}} \cJ$, where $\cJ$ is a complex structure on $\rep{r}$.
In this case $\Sym^2\!\rep{r} \supset \tadj\Lg$.\footnote{$\tadj\Lg$ must occur in $\rep{r}\otimes\rep{r} = \Sym^2{\rep{r}} \oplus \wedge^2\rep{r}$, and since $\tadj\Lg$ is irreducible, it must occur in the first factor, since otherwise $\rep{r}$ would be real.}
\item
$\rep{r}$ is complex if it is neither real or pseudoreal, in which case $\rep{r}\otimes\brep{r} \supset \rep{1} \oplus\tadj\Lg$.
\end{enumerate}
While this is  familiar for $\rep{r}$ an irrep of a simple Lie algebra $\Lg$, it holds more generally for {\it any faithful} irrep of a {\it semisimple} $\Lg$.\footnote{Let $V_{\rep{r}}$ denote the vector space of the irrep $\rep{r}$. A representation $\Lg \to \GL(V_{\rep{r}})$ is faithful if it has a trivial kernel.}  As this is perhaps less familiar, we provide the following lemma.
\begin{lemma}
\label{lem:semisimp}
Let $\rep{r} = (\rep{\rho_1},\rep{\rho_2},\ldots,\rep{\rho_k})$ be a faithful irrep of a semisimple Lie algebra $\Lg = \Lg_1\oplus\Lg_2\oplus\cdots\oplus \Lg_k$.  If $\rep{r}$ is real, then $\Sym^2\!\rep{r} \supset \rep{1}$ and $\wedge^2\rep{r} \supset\tadj\Lg$.  If $\rep{r}$ is pseudoreal, then $\Sym^2\!\rep{r} \supset \tadj \Lg$ and $\wedge^2\rep{r}\supset \rep{1}$. If $\rep{r}$ is complex, then $\rep{r}\otimes\brep{r} \supset \rep{1}\oplus \tadj\Lg$.   Moreover, the statements about the singlets remain true even if $\rep{r}$ fails to be faithful.
\end{lemma}  
\begin{proof}
We will describe the proof for $\rep{r}$ real; the other cases are handled analogously.\\[2mm]  
The symmetric bilinear invariant for $\rep{r}$ must be a tensor product of invariants of the $\rep{\rho_s}$ irreps.  Since each $\rep{\rho_s}$ has at most one invariant that is either symmetric or anti-symmetric, each $\rep{\rho_s}$ must be real or pseudoreal, and for real $\rep{r}$ the number of pseudoreal $\rep{\rho_s}$ must be even.  

The result clearly holds for $k=1$, where $\Lg$ is simple.  Assuming it holds for $\Lg=\Lg_0$, there are two ways to increase $k$:
\vspace{-2mm}
\begin{enumerate}
\item $\Lg = \Lg_0\oplus \Lg_{k+1}$ and $\rep{r} = (\rep{r_0},\rep{\rho^+_{k+1}})$ : 
\vspace{-2mm}
\beq
\begin{split}
\Sym^2\!\rep{r} &\supset (\Sym^2\!\rep{r_0},\Sym^2\!\rep{\rho^+_{k+1}})  \supset (\rep{1},\rep{1})~, \cr
\wedge^2\rep{r} & = (\Sym^2\!\rep{r_0},\wedge^2\rep{\rho^+_{k+1}}) \oplus (\wedge^2\rep{r_0},\Sym^2\!\rep{\rho^+_{k+1}}) \cr
&\supset (\rep{1},\tadj\Lg_{k+1}) \oplus (\tadj\Lg_0,\rep{1})~.
\end{split}\label{test}
\eeq
\item $\Lg = \Lg_0\oplus\Lg_{k+1}\oplus\Lg_{k+2}$ and $\rep{r} = (\rep{r_0},\rep{\rho^-_{k+1}},\rep{\rho^-_{k+2}})$:
\beq
\begin{split}
\Sym^2\!\rep{r} & \supset 
(\Sym^2\!\rep{r_0},\wedge^2 \rep{\rho^-_{k+1}},\wedge^2 \rep{\rho^-_{k+2}}) \supset (\rep{1},\rep{1},\rep{1})~,\cr
\wedge^2\rep{r} & \supset   (\wedge^2 \rep{r_0}, \wedge^2\rep{\rho^-_{k+1}},\wedge^2 \rep{\rho^-_{k+2}})   \oplus (\Sym^2\rep{r_0},\Sym^2\!\rep{\rho^-_{k+1}},\wedge^2 \rep{\rho^-_{k+2}})  \cr
&\qquad\qquad\oplus  (\Sym^2\!\rep{r_0},\wedge^2\rep{\rho^-_{k+1}},\Sym^2\!\rep{\rho^-_{k+2}})
 \cr
&\supset  (\tadj\Lg_0,\rep{1},\rep{1}) \oplus (\rep{1},\tadj\Lg_{k+1}, \rep{1}) \oplus (\rep{1},\rep{1},\tadj\Lg_{k+2})~.
\end{split}
\eeq
\end{enumerate}
The result for faithful $\rep{r}$ follows by induction on $k$. Finally, $\rep{r}$ fails to be faithful if and only if $\rep{\rho_s} = \rep{1}$ for some $s$, in which case $\tadj \Lg_s$ will not show up in the decompositions, but the indicated singlets will still be present.  
\end{proof}

The conjugate representation $\brep{r}$ of a semisimple $\Lg$ is related by a similarity transformation to $\rep{r}$ if and only if $\rep{r}$ is real or pseudoreal.  We see from above that for any irrep $\rep{r}$, $\rep{r}\otimes\brep{r}\supset \rep{1}$.  In fact, using crossing symmetry (i.e.~associativity of the tensor product), we have the following result~\cite{Slansky:1981yr,DiFrancesco:1997nk}: 
\begin{lemma}
\label{lemma:crossing}
 Given two irreps $\rep{r_1}$ and $\rep{r_2}$ of a semisimple Lie algebra $\Lg$, $\rep{r_1} \otimes \rep{r_2} \supset \rep{1}$ if and only if $\rep{r_1} = \brep{r_2}$.
\end{lemma}
The more general statement of crossing symmetry is that if $\rep{r_1}\otimes\rep{r_2} \supset \rep{r_3}$, then $\rep{r_1}\otimes\brep{r_3}$ contains $\brep{r_2}$.  Our result follows by setting $\rep{r_3} = \rep{1}$. 

Having reviewed some basic terminology, we end this section with two results on the branching of pseudoreal representations.
\begin{lemma}
\label{lem:simpseud1}
Let $\rep{R}$ be a pseudoreal irrep of a semisimple Lie algebra $\Lg$, and let $\LLh$ be a semisimple subalgebra of $\Lg$. Then
\begin{align*}
\rep{R} = 
\bigoplus_i (\rep{r^+_i}\oplus \rep{r^+_i} )^{\oplus k_i} \oplus 
\bigoplus_p (\rep{r^-_p})^{\oplus l_p} \oplus
\bigoplus_q (\rep{r^c_q}\oplus \brep{r^c_q})^{\oplus m_q}~,
\end{align*}
where $\rep{r^+_i}$, $\rep{r^-_p}$, and $\rep{r^c_q}$ are distinct real, pseudoreal, and complex irreps of $\LLh$.
\end{lemma}
\begin{proof}
We can decompose $\rep{R}$ as
\begin{align}
\rep{R} = \bigoplus_i (\rep{r^+_i})^{\oplus K_i} \oplus \bigoplus_p (\rep{r^-_p})^{\oplus l_p} \oplus \bigoplus_Q (\rep{r^c_Q})^{\oplus m_Q}~,
\end{align}
where $\rep{r^+_i}$, $\rep{r^-_p}$ and $\rep{r^c_Q}$ are inequivalent irreps.
The generators $\cT_{\rep{R}}$ are  block-diagonal with respect to the decomposition and satisfy 
\begin{align}
\label{eq:pseudo}
\cJ \cT^\ast_{\rep{R}} =  \cT_{\rep{R}} \cJ.
\end{align}
$\cJ$ must act block-diagonally on each block of inequivalent real or pseudoreal representations in the sum.  Furthermore, since $\rep{r^c_Q}$ is not conjugate to $\brep{r^c_Q}$, in order to match the two sides of~(\ref{eq:pseudo}), $\rep{r^c_q}$ occurs in the decomposition only if $\brep{r^c_q}$ occurs as well.  Hence, 
\begin{align}
\rep{R} = \bigoplus_i (\rep{r^+_i})^{\oplus K_i} \oplus \bigoplus_p (\rep{r^-_p})^{l_p} \oplus \bigoplus_q (\rep{r^c_q}\oplus\brep{r^c_q})^{\oplus m_q}~.
\end{align}
Consider the action of $\cJ$ on  $(\rep{r^+_i})^{\oplus K_i}$, denoted by $\cJ_i$.  Without loss of generality the generators $t_i$ of $\LLh$ in $\rep{r^+_i}$ can be taken to be real, and  $\cJ_i = \sum_s M_s \otimes \tau_s$, where $M_s$ is a $K_i\times K_i$ matrix, and $\tau_s$ acts on $\rep{r^+_i}$.  The restriction of~(\ref{eq:pseudo}) to this block is 
\begin{align}
\sum_s M_s \otimes (-t_i \tau_s + \tau_s t_i ) = 0~,
\end{align}
and since $\rep{r^+_i}$ is an irrep, $\cJ_i = M \otimes \iden$ for some invertible skew-symmetric $M$.  Thus, $K_i = 2k_i$, and $M$ is a complex structure on $\C^{k_i}$.  The result follows.
\end{proof}
Analogous considerations determine the action of the complex structure $\cJ$ on the remaining blocks:  $\cJ_p = \iden_{l_p\times l_p} \otimes j_p$, where $j_p$ is the bilinear invariant of $\rep{r^-_p}$, while the action of $\cJ_q$ on $(\rep{r^c_q} \oplus \brep{r^c_q})^{\oplus m_q}$ has the same form as $\cJ_i$, but with $k_i$ replaced by $m_q$.  Hence, we have the following.
\begin{lemma}
\label{lem:simpseud2}
Let $\rep{R}$ be a pseudoreal irrep of a semisimple Lie algebra $\Lg$, and let $\LLh\oplus\LLh'$ be a semisimple subalgebra of $\Lg$. Decomposing $\rep{R}$ with respect to $\LLh\oplus\LLh'$, Lemma~\ref{lem:simpseud1} is refined to
\begin{align*}
\rep{R} = \bigoplus_i (\rep{r^+_i},\rep{R_i}) \oplus \bigoplus_p (\rep{r^-_p},\rep{R_p}) \oplus \bigoplus_q 
\left[(\rep{r^c_q},\rep{R_q})\oplus(\brep{r^c_q},\brep{R_q})\right]~.
\end{align*}
While $\rep{R_i}$, $\rep{R_p}$, $\rep{R_q}$ need not be irreps of $\LLh'$ , $\wedge^2\rep{R_i}\supset\rep{1}$
and $\Sym^2\!\rep{R_p}\supset\rep{1}$.
\end{lemma} 

\subsection{Global symmetries}
We now have the tools to prove Theorem~\ref{thm:fieldsmedal}, and we present the proof in this section.  Let $\Lg \subset \sp(n)$ be a semisimple subalgebra with commutant $\LCg$.  It is easy to show that $\Lg\cap \,\LCg = 0$, so that $\Lg\oplus\,\LCg$ is a subalgebra of $\sp(n)$, and $\LCg$ is reductive, i.e.~a sum $\LCg = \LLh \oplus \Lu(1)^{\oplus A}$ of a semisimple factor $\LLh$ and an abelian factor.  Using Lemma~\ref{lem:simpseud2}, we decompose $\rep{2n}$ as
\begin{align}
\rep{2n} = \bigoplus_{i} (\rep{r^+_i}, \rep{R_i}) 
\oplus \bigoplus (\rep{r^-_p}, \rep{R_p}) \oplus 
\bigoplus_{q} (\rep{r^c_q},\rep{R_q})\oplus(\brep{r^c_q},\brep{R_q})~,
\end{align}
where $\rep{r^+_i}$, $\rep{r^-_p}$ and $\rep{r^c_q}$ denote distinct irreps of $\Lg$ with indicated reality properties.
Since $\tadj \sp(n) = \Sym^2\!\rep{2n}$, we find
\begin{align}
\tadj\sp(n) & \supset 
\bigoplus_{i} (\Sym^2\!\rep{r^+_i}, \Sym^2\!\rep{R_i})  \oplus
\bigoplus_p (\wedge^2 \rep{r^-_p}, \wedge^2\rep{R_p}) \oplus 
\bigoplus_q (\rep{r^c_q}\otimes\brep{r^c_q}, \rep{R_q}\otimes\brep{R_q}) \nonumber\\
& \supset 
\bigoplus_{i} (\rep{1}, \Sym^2\!\rep{R_i})  \oplus
\bigoplus_p (\rep{1}, \wedge^2\rep{R_p}) \oplus 
\bigoplus_q (\rep{1}, \rep{R_q}\otimes\brep{R_q})~.
\end{align}
By Lemma~\ref{lemma:crossing} every $\Lg$-singlet in $\tadj\sp(n)$ is obtained this way, and by assumption these $\Lg$ singlets are precisely the generators of $\LCg$.  Decomposing further into irreps of $\LLh$ as
\begin{align}
\rep{R_i} &= \bigoplus_\alpha \rep{\rho_{i\alpha}}~,&
\rep{R_p} & = \bigoplus_\sigma \rep{\rho_{p\sigma}}~,&
\rep{R_q} & = \bigoplus_\mu \rep{\rho_{q\mu}}~,
\end{align}
we obtain
\begin{align}
\label{eq:irreps}
\tadj\LLh \oplus\Lu(1)^A &  = 
\bigoplus_{i} \bigoplus_{\alpha} \Sym^2\!\rep{\rho_{i\alpha} } 
\oplus
\bigoplus_{p} \bigoplus_\sigma \wedge^2  \rep{\rho_{p\sigma}}
\oplus
\bigoplus_{q}  \bigoplus_{\mu} \rep{\rho_{q\mu}}\otimes\brep{\rho_{q\mu}} \nonumber\\
&\quad\oplus
\bigoplus_{i} \bigoplus_{\alpha>\beta} \rep{ \rho_{i\alpha}}\otimes\rep{\rho_{i\beta}}
\oplus
\bigoplus_{p} \bigoplus_{\sigma>\tau} \rep{\rho_{p\sigma}}\otimes\rep{\rho_{p\tau}}
\oplus\bigoplus_{q}  \bigoplus_{\mu\neq\nu} \rep{\rho_{q\mu}}\otimes\brep{\rho_{q\nu}}~~.
\end{align}
Decomposing $\LLh = \oplus_s \LLh_s$ into its simple summands, we observe that every summand in
\begin{align}
\label{eq:adjs}
\tadj\LLh = (\tadj\LLh_1,\rep{1},\ldots,\rep{1}) \oplus (\rep{1},\tadj\LLh_2,\ldots,\rep{1})\oplus\cdots\oplus(\rep{1},\ldots,\rep{1},\tadj\LLh_k)
\end{align}
must be contained in exactly one of the summands in~(\ref{eq:irreps}), in fact a summand on the first line of~(\ref{eq:irreps}).\footnote{  To see the latter, assume the contrary, {\it e.g.}~$(\tadj\LLh_1,\rep{1},\ldots,\rep{1}) \subset\rep{ \rho_{i\alpha}}\otimes\rep{\rho_{i\beta}}$ for $\alpha\neq \beta$.  This can only work if $\rep{\rho_{i\alpha}}$ or $\rep{\rho_{i\beta}}$ is non-trivial, but in that case $\Sym^2\!\rep{\rho_{i\alpha}}$ or $\Sym^2\!\rep{\rho_{i\beta}}$ will yield additional terms in the decomposition.  Similar reasoning excludes the other summands in the second line of~(\ref{eq:irreps}).}  The second line must be absent, i.e., $\rep{R_i}$, $\rep{R_p}$ and $\rep{R_q}$ must in fact be irreps of $\LLh$; otherwise the right-hand side of~(\ref{eq:irreps}) would necessarily contain extra non-trivial representations of $\LLh$.   For the same reason each simple factor $\LLh_s$ must act non-trivially on exactly one of $\rep{R_i}$, $\rep{R_p}$, $\rep{R_q}$, i.e.,
\begin{align}
\bigoplus_s \LLh_s = 
\bigoplus_i \LLh_i \oplus
\bigoplus_p \LLh_p \oplus
\bigoplus_q \LLh_q~,
\end{align}
with
\begin{align}
\tadj \LLh_i &= \Sym^2\!\rep{R_i}~,&
\tadj \LLh_p & = \wedge^2\rep{R_p}~,&
\Lu(1)^{\oplus A} \oplus \bigoplus_q \tadj \LLh_q  & = \bigoplus_q \rep{R_q}\otimes\brep{R_q}~.
\end{align}
We recognize the classical groups $\LLh_i = \sp(k_i)$, 
$\LLh_p = \so(l_p)$, and $\LLh_q = \su(m_q)$, with $\rep{R_i}$, $\rep{R_p}$,  and $\rep{R_q}$ the corresponding fundamental representations.  Moreover, the abelian factor $\Lu(1)^{\oplus A} = \oplus_q \Lu(1)_q$, and $\Lu(1)_q$ acts
with charge $+1$ on $\rep{R_q}$ and $-1$ on $\brep{R_q}$.  This completes the proof of Theorem~\ref{thm:fieldsmedal}.

\section{Discussion}
\label{sec:disc}

Having found that gauging a subalgebra of $\sp(n)$ does not yield $\su(m)$ factors without accompanying $\uu(1)$'s, we now turn to the question of whether it is possible to get such factors in some other way. In particular, we consider two possibilities: gauging discrete subgroups, as well as moving out on the Higgs branch. We will find that discrete gaugings do not yield $\su(m)$'s, whereas special loci on the baryonic branch of SQCD do. Of course, we also can not rule out the possibilities of emergent (accidental) symmetries yielding $\su(m)$ factors, and we will have nothing further to say about this possibility here. 

\subsection{Discrete gauge symmetries}
One way to decrease the global symmetry group $G$ is to introduce a further gauging by a discrete subgroup $\Gamma \subset G$.  The remaining global symmetry will be the centralizer of $\Gamma$ in $G$, $C_\Gamma$.  While this can have interesting consequences for the non-abelian part of the global symmetry, since $C_\Gamma$ will inevitably contain the center of $G$, it cannot affect the abelian $\oplus_q \Lu(1)_q$ component of the symmetry algebra.

\subsection{Higgs branch}

The moduli space of $\cN=2$ $\SU(N_c)$ SQCD with $N_f$ flavors was comprehensively analyzed in \cite{Argyres:1996eh}. In this work, the authors describe the remaining global symmetries on the various possible sub-branches of the Higgs branch.  When $N_c \leq N_f < 2N_c$ the remaining non-R global symmetry on the baryonic branch is $\SU(2N_c - N_f) \times \GU(1)^{N_f - N_c}$. When $N_f = N_c$, the $\GU(1)$ factors are spontaneously broken, and the global symmetry is simply $\SU(N_f)$. Moreover, even when $N_f > N_c$, the $\GU(1)$ factors do not enhance $\SU(2N_c - N_f)$ to $\GU(2N_c - N_f)$.  Thus it is possible to get non-enhanced $\SU(m)$ factors on the Higgs branch of $\cN=2$ theories.

\subsection{General discussion and conclusions}
 Let us now comment on some special cases of interest, in particular those of the low-rank $T_N$ theories. The first non-trivial case is the $T_2$ theory. This has a na\"{i}ve global symmetry algebra $\su(2)^{\oplus 3}$ and is known to be equivalent to a free theory of $8$ chiral multiplets transforming in the tri-fundamental representation of $\su(2)^{\oplus 3}$.   From the perspective of the analysis in section~\ref{sec:free} it is clear that the global symmetry algebra is $\sp(4)$, and under $\sp(4) \supset \su(2)^{\oplus 3}$ the matter decomposes as $\rep{8} = (\rep{2},\rep{2},\rep{2})$.
 
The $T_3$ theory has a similar structure. Na\"{i}vely this theory has a global symmetry algebra $\su(3)^{\oplus 3}$ with chiral multiplets transforming in the tri-fundamental $(\rep{3},\rep{3},\rep{3})$. In fact, it is enhanced to $\eu{e}_6$~\cite{Gaiotto:2009we}, and the representation theory works out nicely: there is a maximal embedding $\su(3)^{\oplus 3}\subset \mathfrak{e}_6$ under which
\begin{align}
\rep{78} =  (\rep{8},\rep{1},\rep{1})\oplus(\rep{1},\rep{8},\rep{1})\oplus(\rep{1},\rep{1},\rep{8})\oplus (\rep{3},\rep{3},\rep{3})\oplus (\brep{3},\brep{3},\brep{3})~.
\end{align}
In other words, the trifundamental fields are additional global currents that enhance the naive $\su(3)^{\oplus 3}$ to $\Le_6$.

Finally, consider the $T_4$ theory with its global symmetry algebra $\su(4)^{\oplus 3}\cong \so(6)^{\oplus 3}$ and matter in $(\rep{4},\rep{4},\rep{4})$.  At first glance one might hope that here a simple weakly-coupled UV Lagrangian is not ruled out by our results, since of course we can easily construct an $\so(6)^{\oplus 3}$ symmetry algebra.  Alas, the hope is short-lived---in a theory so obtained the matter would transform in $\rep{6}$ for each of the $\so(6)$ factors, and no tensor product could produce the desired $\rep{4}$ spinor representations.

We now conclude with a few brief comments. Although it is too strong to say that we have proven that  $T_N$ theories do not arise via gauging the symmetries of free hypermultiplets, we have ruled out the simplest realizations that do not explore the Higgs branch of the $\cN=2$ gauge theory.  Consider moving out onto the Higgs branch by giving a field a vev $v$, and let the strong coupling scale of the UV gauge theory be denoted by $\Lambda$. If $v \gg \Lambda$, the IR gauge-neutral degrees of freedom, whose vevs parametrize the flat directions, will decouple from the IR gauge sector.  The symmetries of the IR gauge theory will then again be constrained by Theorem~\ref{thm:fieldsmedal}.  If, on the other hand, $v\sim \Lambda$, then the dynamics is necessarily strongly coupled and outside of the domain of validity of our results.  

Of course a Lagrangian realization for $T_N$ has long been suspected to be highly unlikely, in light of the poorly understood dynamics of the M5-brane origin of such theories; for example, the $N^3$ scaling of the number of degrees of freedom in these systems does not seem to have any obvious gauge theory realization. Moreover, the $T_N$ theories have no marginal deformations, so they do not seem to arise as SCFTs in the same way as $N_f = 2N_c$ gauge theories, which have an exactly marginal gauge coupling. 

However, even aside from possible applications to strongly coupled theories, our main result indicates just how strongly constrained the global symmetries of $\cN=2$ gauge theories are and will perhaps provide a useful step towards a classification of such theories.  For instance, by combining our results with the recent work~\cite{Bhardwaj:2013qia}, it should be easy to give a comprehensive list of all possible symmetry algebras of conformal and asymptotically free theories.  It would be interesting to extend that to include possible discrete gaugings. It may perhaps also be useful to extend our results to $\cN=1$ theories as well, though there we expect important new complications from possible superpotential interactions.

\newpage

\begin{center}
\bf{Acknowledgements}
\end{center}
\medskip
We would like to thank Ibrahima Bah, Jacques Distler, Ken Intriligator, and David Tong for useful discussions. IVM is supported by the NSF Focused Research Grant DMS-1159404  and Texas A\&M. BW is supported in part by the STFC Standard Grant ST/J000469/1 ``String Theory, Gauge Theory and Duality."  IVM and BW would like to thank the organizers of SMUK'13, where this collaboration began.  JMO and BW would like to thank the Albert Einstein Institute for hospitality while this work was undertaken, and JMO would like to thank Texas A\&M for hospitality while this work was being completed.

\vspace*{0.2in}


\begin{thebibliography}{10}

\bibitem{Minahan:1996cj}
J.~A. Minahan and D.~Nemeschansky, ``{Superconformal fixed points with E(n)
  global symmetry},''
  \href{http://dx.doi.org/10.1016/S0550-3213(97)00039-4}{{\em Nucl.Phys.} {\bf
  B489} (1997)  24--46},
\href{http://arxiv.org/abs/hep-th/9610076}{{\tt arXiv:hep-th/9610076
  [hep-th]}}.

\bibitem{Minahan:1996fg}
J.~A. Minahan and D.~Nemeschansky, ``{An N=2 superconformal fixed point with
  E(6) global symmetry},''
  \href{http://dx.doi.org/10.1016/S0550-3213(96)00552-4}{{\em Nucl.Phys.} {\bf
  B482} (1996)  142--152},
\href{http://arxiv.org/abs/hep-th/9608047}{{\tt arXiv:hep-th/9608047
  [hep-th]}}.

\bibitem{Seiberg:1994rs}
N.~Seiberg and E.~Witten, ``{Electric - magnetic duality, monopole
  condensation, and confinement in N=2 supersymmetric Yang-Mills theory},''
  \href{http://dx.doi.org/10.1016/0550-3213(94)90124-4}{{\em Nucl.Phys.} {\bf
  B426} (1994)  19--52},
\href{http://arxiv.org/abs/hep-th/9407087}{{\tt arXiv:hep-th/9407087
  [hep-th]}}.

\bibitem{Seiberg:1994aj}
N.~Seiberg and E.~Witten, ``{Monopoles, duality and chiral symmetry breaking in
  N=2 supersymmetric QCD},'' {\em Nucl. Phys.} {\bf B431} (1994)  484--550,
\href{http://arxiv.org/abs/hep-th/9408099}{{\tt arXiv:hep-th/9408099}}.

\bibitem{Gaiotto:2008nz}
D.~Gaiotto, A.~Neitzke, and Y.~Tachikawa, ``{Argyres-Seiberg duality and the
  Higgs branch},'' \href{http://dx.doi.org/10.1007/s00220-009-0938-6}{{\em
  Commun.Math.Phys.} {\bf 294} (2010)  389--410},
\href{http://arxiv.org/abs/0810.4541}{{\tt arXiv:0810.4541 [hep-th]}}.

\bibitem{Argyres:2007cn}
P.~C. Argyres and N.~Seiberg, ``{S-duality in N=2 supersymmetric gauge
  theories},'' \href{http://dx.doi.org/10.1088/1126-6708/2007/12/088}{{\em
  JHEP} {\bf 0712} (2007)  088},
\href{http://arxiv.org/abs/0711.0054}{{\tt arXiv:0711.0054 [hep-th]}}.

\bibitem{Argyres:2007tq}
P.~C. Argyres and J.~R. Wittig, ``{Infinite coupling duals of N=2 gauge
  theories and new rank 1 superconformal field theories},''
  \href{http://dx.doi.org/10.1088/1126-6708/2008/01/074}{{\em JHEP} {\bf 0801}
  (2008)  074},
\href{http://arxiv.org/abs/0712.2028}{{\tt arXiv:0712.2028 [hep-th]}}.

\bibitem{Argyres:1995jj}
P.~C. Argyres and M.~R. Douglas, ``{New phenomena in SU(3) supersymmetric gauge
  theory},'' \href{http://dx.doi.org/10.1016/0550-3213(95)00281-V}{{\em
  Nucl.Phys.} {\bf B448} (1995)  93--126},
\href{http://arxiv.org/abs/hep-th/9505062}{{\tt arXiv:hep-th/9505062
  [hep-th]}}.

\bibitem{Gaiotto:2009we}
D.~Gaiotto, ``{N=2 dualities},''
  \href{http://dx.doi.org/10.1007/JHEP08(2012)034}{{\em JHEP} {\bf 1208} (2012)
   034},
\href{http://arxiv.org/abs/0904.2715}{{\tt arXiv:0904.2715 [hep-th]}}.

\bibitem{Tachikawa:2009rb}
Y.~Tachikawa, ``{Six-dimensional D(N) theory and four-dimensional SO-USp
  quivers},'' \href{http://dx.doi.org/10.1088/1126-6708/2009/07/067}{{\em JHEP}
  {\bf 0907} (2009)  067},
\href{http://arxiv.org/abs/0905.4074}{{\tt arXiv:0905.4074 [hep-th]}}.

\bibitem{Chacaltana:2010ks}
O.~Chacaltana and J.~Distler, ``{Tinkertoys for Gaiotto Duality},''
  \href{http://dx.doi.org/10.1007/JHEP11(2010)099}{{\em JHEP} {\bf 1011} (2010)
   099},
\href{http://arxiv.org/abs/1008.5203}{{\tt arXiv:1008.5203 [hep-th]}}.

\bibitem{Chacaltana:2011ze}
O.~Chacaltana and J.~Distler, ``{Tinkertoys for the $D_N$ series},''
\href{http://arxiv.org/abs/1106.5410}{{\tt arXiv:1106.5410 [hep-th]}}.

\bibitem{Benini:2009mz}
F.~Benini, Y.~Tachikawa, and B.~Wecht, ``{Sicilian gauge theories and N=1
  dualities},'' \href{http://dx.doi.org/10.1007/JHEP01(2010)088}{{\em JHEP}
  {\bf 1001} (2010)  088},
\href{http://arxiv.org/abs/0909.1327}{{\tt arXiv:0909.1327 [hep-th]}}.

\bibitem{Argyres:1996eh}
P.~C. Argyres, M.~R. Plesser, and N.~Seiberg, ``{The Moduli space of vacua of
  N=2 SUSY QCD and duality in N=1 SUSY QCD},''
  \href{http://dx.doi.org/10.1016/0550-3213(96)00210-6}{{\em Nucl.Phys.} {\bf
  B471} (1996)  159--194},
\href{http://arxiv.org/abs/hep-th/9603042}{{\tt arXiv:hep-th/9603042
  [hep-th]}}.

\bibitem{Argyres:1996hc}
P.~C. Argyres, M.~R. Plesser, and A.~D. Shapere, ``{N=2 moduli spaces and N=1
  dualities for SO(n(c)) and USp(2n(c)) superQCD},''
  \href{http://dx.doi.org/10.1016/S0550-3213(96)00583-4}{{\em Nucl.Phys.} {\bf
  B483} (1997)  172--186}, \href{http://arxiv.org/abs/hep-th/9608129}{{\tt
  arXiv:hep-th/9608129 [hep-th]}}.

\bibitem{Cahn:1984la}
R.~Cahn, {\em Semi-simple {Lie} algebras and their representations}.
\newblock Benjaming Cummings, 1985.

\bibitem{McKay:1981rt}
W.~G. McKay and J.~Patera, {\em Tables of dimensions, indices, and branching
  rules for representations of simple {L}ie algebras}, vol.~69 of {\em Lecture
  Notes in Pure and Applied Mathematics}.
\newblock Marcel Dekker Inc., New York, 1981.

\bibitem{Slansky:1981yr}
R.~Slansky, ``{Group Theory for Unified Model Building},''
\href{http://dx.doi.org/10.1016/0370-1573(81)90092-2}{{\em Phys.Rept.} {\bf 79}
  (1981)  1--128}.

\bibitem{Zumino:1962nf}
B.~Zumino, ``Normal forms of complex matrices,'' {\em J. Math. Phys.} {\bf 3}
  (1962)  .

\bibitem{DiFrancesco:1997nk}
P.~Di~Francesco, P.~Mathieu, and D.~Senechal, {\em {Conformal field theory}}.
\newblock Springer,
1997.
\newblock

\bibitem{Bhardwaj:2013qia}
L.~Bhardwaj and Y.~Tachikawa, ``{Classification of 4d N=2 gauge theories},''
  \href{http://arxiv.org/abs/1309.5160}{{\tt arXiv:1309.5160 [hep-th]}}.

\end{thebibliography}

\providecommand{\href}[2]{#2}\begingroup\raggedright\endgroup
\end{document}